\documentclass[letterpaper, 10 pt, conference]{ieee_templates/ieeeconf}
\IEEEoverridecommandlockouts 
\usepackage{amsmath,amsfonts,amssymb,amsthm}
\usepackage{wrapfig}
\usepackage[linesnumbered,ruled,vlined]{algorithm2e}
\newtheorem{theorem}{Theorem}[section]
\newtheorem{definition}{Definition}[section]

\usepackage{graphicx}
\usepackage{float}
\usepackage{microtype}
\usepackage{graphicx}
\usepackage{color}
\usepackage{hyperref}
\usepackage[percent]{overpic}

\usepackage[font=small,skip=1pt]{caption}

\title{\LARGE \bf 
Bayesian Learning-Based Adaptive Control for Safety Critical Systems
}

\author{David D. Fan$^{1,3}$, Jennifer Nguyen$^{2}$, Rohan Thakker$^{3}$, Nikhilesh Alatur$^{3}$,\\ Ali-akbar Agha-mohammadi$^{3}$,
and Evangelos A. Theodorou$^{1}$
\thanks{$^{1}$Institute for Robotics and Intelligent Machines, Georgia Institute of Technology, Atlanta, GA, USA}%
\thanks{$^{2}$Department of Mechanical and Aerospace Engineering, West Virginia University, Morgantown, WV, USA}%
\thanks{$^{3}$NASA Jet Propulsion Laboratory, California Institute of Technology, Pasadena, CA, USA}%
}

\begin{document}
\bstctlcite{IEEEexample:BSTcontrol}
\maketitle


\begin{abstract}
Deep learning has enjoyed much recent success, and applying state-of-the-art model learning methods to controls is an exciting prospect.  However, there is a strong reluctance to use these methods on safety-critical systems, which have constraints on safety, stability, and real-time performance.  We propose a framework which satisfies these constraints while allowing the use of deep neural networks for learning model uncertainties.  Central to our method is the use of Bayesian model learning, which provides an avenue for maintaining appropriate degrees of caution in the face of the unknown.  In the proposed approach, we develop an adaptive control framework leveraging the theory of stochastic CLFs (Control Lyapunov Functions) and stochastic CBFs (Control Barrier Functions) along with tractable Bayesian model learning via Gaussian Processes or Bayesian neural networks.  Under reasonable assumptions, we guarantee stability and safety while adapting to unknown dynamics with probability 1.  We demonstrate this architecture for high-speed terrestrial mobility targeting potential applications in safety-critical high-speed Mars rover missions.
\end{abstract}



\section{Introduction}

The rapid growth of Artificial Intelligence (AI) and Machine Learning (ML) disciplines has created a tremendous impact in engineering disciplines, including finance, medicine, and general cyber-physical systems. The ability of ML algorithms to learn high dimensional dependencies has expanded the capabilities of traditional disciplines and opened up new opportunities towards the development of decision making systems which operate in complex scenarios.  Despite these recent successes \cite{silver2017mastering}, there is low acceptance of AI and ML algorithms to safety-critical domains, including human-centered robotics, and particularly in the flight and space industries.  For example, both recent and near-future planned Mars rover missions largely rely on daily human decision making and piloting, due to a very low acceptable risk for trusting black-box autonomy algorithms.  Therefore there is a need to develop computational tools and algorithms that bridge two worlds: the canonical structure of control theory, which is important for providing guarantees in safety-critical applications, and the data driven abstraction and representational power of machine learning, which is necessary for adapting the system to achieve resiliency against unmodeled disturbances.

\begin{figure}
\centering
\includegraphics[width=0.44\linewidth]{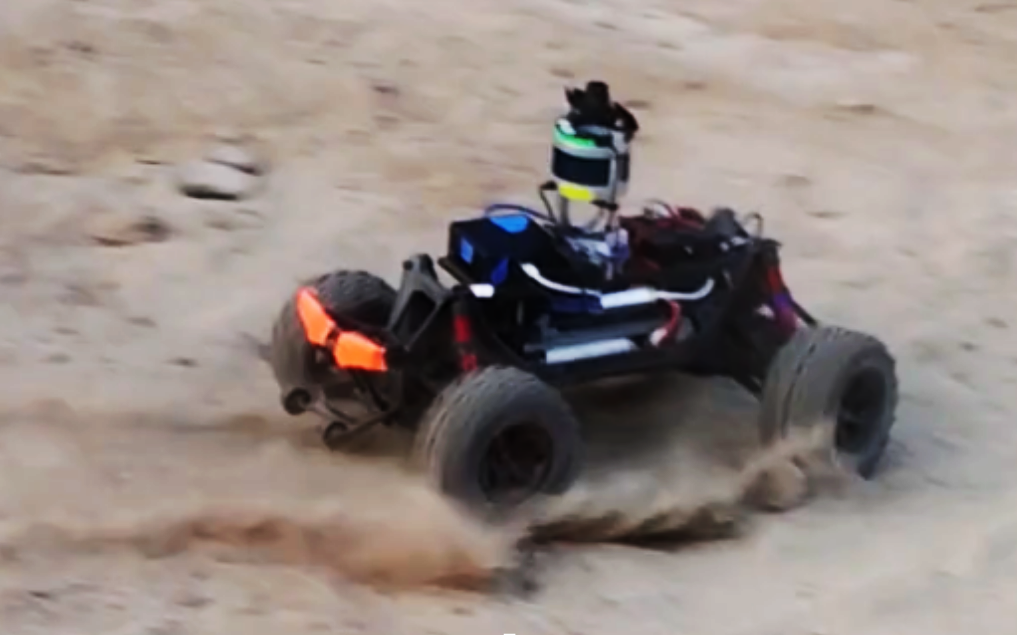}
\includegraphics[width=0.503\linewidth]{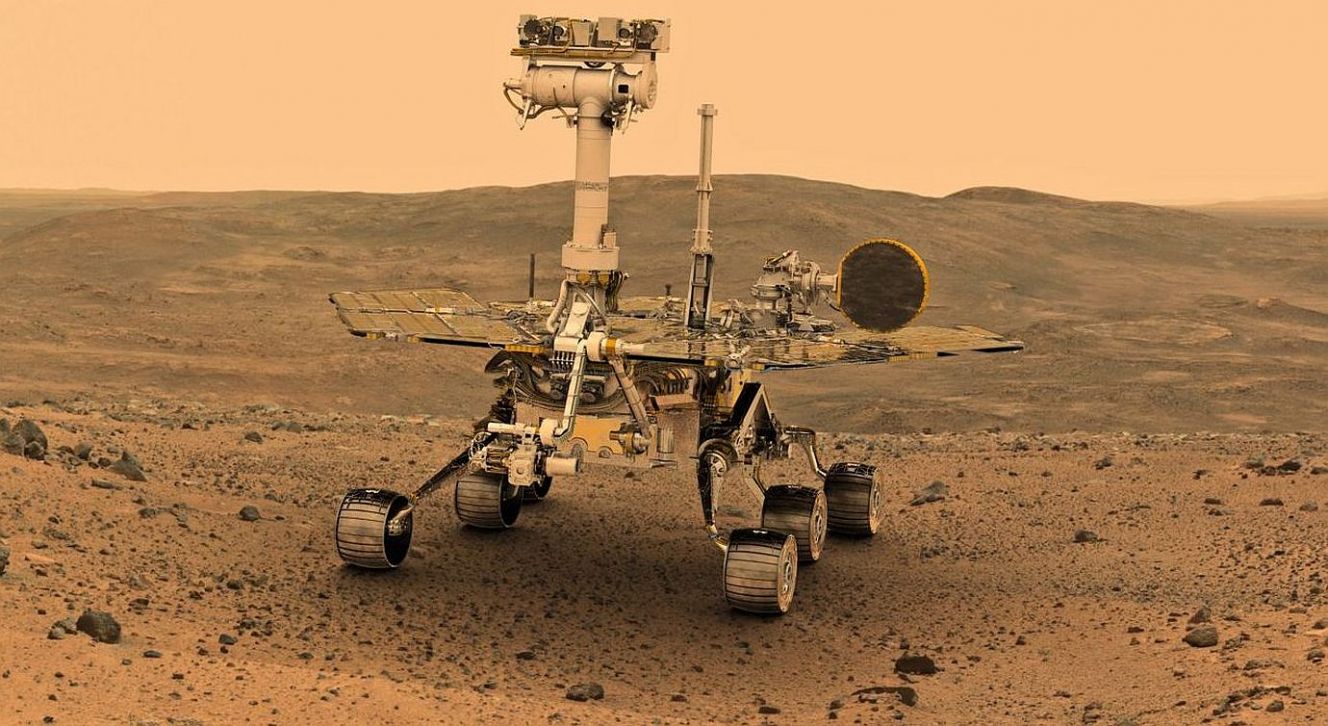}
\caption{The left image depicts a 1/5th scale RC car platform driving at the Mars Yard at JPL; and the right is a platform from the Mars Explore Rover (MER) mission. }
\label{fig:car}
\end{figure}

Towards this end, we propose a novel, lightweight framework for Bayesian adaptive control for safety critical systems, which we call BALSA (BAyesian Learning-based Safety and Adaptation).  This framework leverages ML algorithms for learning uncertainty representations of dynamics which in turn are used to generate sufficient conditions for stability using stochastic CLFs and safety using stochastic CBFs.  Treating the problem within a stochastic framework allows for a cleaner and more optimal approach to handling modeling uncertainty, in contrast to deterministic, discrete-time, or robust control formulations.  We apply our framework to the problem of high-speed agile autonomous vehicles, a domain where learning is especially important for dynamics which are complex and difficult to model (e.g., fast autonomous driving over rough terrain).  Potential Mars Sample Return (MSR) missions are one example in this domain.  Current Mars rovers (i.e., Opportunity and Curiosity) have driven on average $\sim$3km/year \cite{nasa_curiosity,nasa_opportunity}. In contrast, if MSR launches in 2028, then the rover has only 99 sols ($\sim$102 days) to complete potentially 10km \cite{2014_klein,  2019_nelessen}.  After factoring in the intermittent and heavily delayed communications to earth, the need for \textit{adaptive}, high-speed autonomous mobility could be crucial to mission success.

Along with the requirements for safety and adaptation, computational efficiency is of paramount importance for real systems.  Hardware platforms often have severe power and weight requirements, which significantly reduce their computational power.  Probabilistic learning and control over deep Bayesian models is a computationally intensive problem.  In contrast, we shorten the planning horizon and rely on a high-level, lower fidelity planner to plan desired trajectories.  Our method then guarantees safe trajectory tracking behavior, even if the given trajectory is not safe.  This frees up the computational budget for other tasks, such as online model training and inference.


\textbf{Related work} - 
Machine-learning based planning and control is a quickly growing field.  From Model Predictive Control (MPC) based learning \cite{wagner,williams2018information}, safety in reinforcement learning \cite{berkenkamp2017safe}, belief-space learning and planning \cite{kim2019bi}, to imitation learning \cite{ross2011reduction}, these approaches all demand considerations of safety under learning \cite{Ostafew2016a, Pereida2018, Hewing2017, Shi2018}.
Closely related to our work is Gaussian Process-based Bayesian Model Reference Adaptive Control (GP-MRAC) \cite{Chowdhary2015}, where modeling error is approximated with a Gaussian Process (GP).  However, computational speed of GPs scales poorly with the amount of data ($\mathcal{O}(N^3)$), and sparse approximations lack representational power.  Another closely related work is that of \cite{Nguyen}, who showed how to formulate a robust CLF which is tolerant to bounded model error.  Extensions to robust CBFs were given in \cite{nguyen2016optimal}.  A stated drawback of this approach is the conservative nature of the bounds on the model error.  In contrast, we incorporate model learning into our formulation, which allows for more optimal behavior, and leverage stochastic CLF and CBF theory to guarantee safety and stability with probability 1.  Other related works include \cite{Cheng2019}, which uses GPs in CBFs to learn the drift term in the dynamics $f(x)$, but uses a discrete-time, deterministic formulation.  \cite{Nguyen2015} combined L1 adaptive control and CLFs.  Learning in CLFs and CBFs using adaptive control methods (including neuro-adaptive control) has been considered in several works, e.g. \cite{Taylor2019,gurriet2018towards,azimi2018robust,azimi2018performance}.

\textbf{Contributions} - 
Here we take a unique approach to address the aforementioned issues, with the requirements of 1) adaptation to changes in the environment and the system, 2) adaptation which can take into account high-dimensional data, 3) guaranteed safety during adaptation, 4) guaranteed stability during adaptation and convergence of tracking errors, 5) low computational cost and high control rates.  Our contributions are fourfold:  First, we introduce a Bayesian adaptive control framework which explicitly uses the model uncertainty to guarantee stability, and is agnostic to the type of Bayesian model learning used.  Second, we extend recent stochastic safety theory to systems with \textit{switched} dynamics to guarantee safety with probability 1.  In contrast to adaptive control, switching dynamics are used to account for model updates which may only occur intermittently.  Third, we combine these approaches in a novel online-learning framework (BALSA).  Fourth, we compare the performance of our framework using different Bayesian model learning and uncertainty quantification methods.  Finally, we apply this framework to a high-speed driving task on rough terrain using an Ackermann-steering vehicle and validate our method on both simulation and hardware experiments.


\section{Safety and Stability under Model Learning via Stochastic CLF/CBFs}

Consider a stochastic system with SDE (stochastic differential equation) dynamics:
\begin{align}
\text{d}x_1=x_2\text{d}t, \quad  \text{d}x_2=(f(x) + g(x)u)\text{d}t + \Sigma(x)\text{d}\xi(t)
\label{eq:dynamics}
\end{align}
where $x_1,x_2\in\mathbb{R}^n$, $x=[x_1,x_2]^\intercal$, the controls are $u\in\mathbb{R}^n$, the diffusion is $\Sigma(x)\in\mathbb{R}^{n\times n}$, and $\xi(t)\in\mathbb{R}^n$ is a zero-mean Wiener process.  For simplicity we restrict our analysis to systems of this form, but emphasize that our results are extensible to systems of higher relative degree \cite{Nguyen2016}, as well as hybrid systems with periodic orbits \cite{ames2014rapidly}.   
A wide range of nonlinear control-affine systems in robotics can be transformed into this form.  In general, on a real system, $f$, $g$, and $\Sigma$ may not be fully known.  We assume $g(x)$ is known and invertible, which makes the analysis more tractable.  It will be interesting in future work to extend our approach to unknown, non-invertible control gains, or non-control affine systems (e.g. $\dot{x}=f(x,u)$).  Let $\hat{f}(x)$ be a given approximate model of $f(x)$.  We formulate a pre-control law with pseudo-control $\mu\in\mathbb{R}^n$:
\begin{equation}
\label{eq:precontrol}
    u = g(x)^{-1}(\mu - \hat{f}(x)) 
\end{equation}
which leads to the system dynamics being
\begin{equation}
    \text{d}x_1 = x_2\text{d}t, \quad \text{d}x_2 = (\mu + \Delta(x))\text{d}t + \Sigma(x)\text{d}\xi(t)
    \label{eq:mrac}
\end{equation}
where $\Delta(x) = f(x)-\hat{f}(x)$ is the modeling error, with $\Delta(x)\in\mathbb{R}^n$.

Suppose we are given a reference model and reference control from, for example, a path planner: 
\begin{align*}
    \text{d}x_{1rm}=x_{2rm}\text{d}t, \quad  \text{d}x_{2rm}=f_{rm}(x_{rm},u_{rm})\text{d}t
\end{align*}
The utility of the methods outlined in this work is for adaptive tracking of this given trajectory with guaranteed safety and stability.  We assume that $f_{rm}$ is continuously differentiable in $x_{rm}$, $u_{rm}$. Further, $u_{rm}$ is bounded and piecewise continuous, and that $x_{rm}$ is bounded for a bounded $u_{rm}$.  Define the error $e=x-x_{rm}$.  We split the pseudo-control input into four separate terms:
\begin{equation}
    \mu=\mu_{rm}+\mu_{pd} + \mu_{qp} -\mu_{ad}
\end{equation}
where we assign $\mu_{rm}=\dot{x}_{2rm}$ and $\mu_{pd}$ to a PD controller:
\begin{equation}
    \mu_{pd} = [-K_P -K_D] e
    \label{eq:pd}
\end{equation}
Additionally, we assign $\mu_{qp}$ as a pseudo-control which we optimize for and $\mu_{ad}$ as an adaptive element which will cancel out the model error.
Then we can write the dynamics of the model error $e$ as:
\begin{align}
    \text{d}e & = \begin{bmatrix}\text{d}e_1 \\ \text{d}e_2\end{bmatrix} = \begin{bmatrix}0 & I \\ -K_P & -K_D\end{bmatrix}e\text{d}t \\
    & \quad \quad + \begin{bmatrix}0 \\ I \end{bmatrix}\big((\mu_{qp} - \mu_{ad} + \Delta(x))\text{d}t + \Sigma(x)\text{d}\xi(t)\big)  \nonumber \\
    &= (Ae + G(\mu_{qp} - \mu_{ad} + \Delta(x))) \text{d}t + G\Sigma(x)\text{d}\xi(t)
    \label{eq:err_deterministic}
\end{align}
where the matrices $A$ and $G$ are used for ease of notation.  The gains $K_D,K_P$ should be chosen such that $A$ is Hurwitz.  When $\mu_{ad} = \Delta(x)$, the drift modeling error term is canceled out from the error dynamics.

Next, we require a method for learning or approximating the drift and diffusion terms $\Delta(x)$ and $\Sigma(x)$.  Such methods include Bayesian SDE approximation methods \cite{look2019differential}, Neural-SDEs \cite{liu2019neural}, or differential GP flows \cite{hegde2019deep}, to name a few.  This model should \textit{know what it doesn't know} \cite{li2011knows}, and should capture both the \textit{epistemic} uncertainty of the model, i.e., the uncertainty from lack of data, as well as the \textit{aleatoric} uncertainty, i.e., the uncertainty inherent in the system \cite{Roy2011}.  We expect that these methods will continue to be improved by the community.  We can use the second equation in (\ref{eq:mrac}) to generate data points to use for learning these terms in the SDE.  In discrete time, the learning problem is formulated as finding a mapping from input data $\bar{X}_t = x(t)$ to output data $\bar{Y}_t=(x_2(t+dt) - x_2(t))/dt - (\hat{f}(x(t)) + g(x(t))u(t))$. Given the $i^{th}$ dataset $\mathcal{D}_i=\{\bar{X}_t,\bar{Y}_t\}_{t=0, dt, \dots, t_i}$ with $i\in\mathbb{N}$, we can construct the $i^{th}$ model $\{m_i(x), \sigma_i(x)\}$, where $m_i(x)$ approximates the drift term $\Delta(x)$ and $\sigma_i(x)$ approximates the diffusion term $\Sigma(x)$.  Note that we do not require updating the model at each timestep, which significantly reduces computational load requirements and allows for training more expressive models (e.g., neural networks).  

In practical terms, in this work we opt for an approximate method for learning $\{m_i(x), \sigma_i(x)\}$, in which we view each data point in $\mathcal{D}_i$ as an independently and identically distributed sample, and set up a single timestep Bayesian regression problem, in which we model $\Delta(x)$ as a multivariate Gaussian random variable, i.e. $\bar{\Delta}_i(x)\sim \mathcal{N}(m_i(x),\sigma_i(x))$.  This approximation ignores the SDE nature of (\ref{eq:mrac}) and will not be a faithful approximation (See \cite{lew2020unpublished} for insightful comments on this problem).  However, until Bayesian SDE approximation methods improve, we believe this approach to be reasonable in practice.  Methods for producing reliable confidence bounds include a large class of Bayesian neural networks (\cite{Hafner2018,Harrison2018,gal2016dropout}), Gaussian Processes or its many approximate variants (\cite{shahriari2015taking,pan2017prediction}), and many others.  We compare several methods in our experimental results.  We leave a more principled learning approach using Bayesian SDE learning methods for future work.

After obtaining the joint model $\{m_i(x), \sigma_i(x)\}$, Equation (\ref{eq:err_deterministic}) can be written as the following switching SDE:
\begin{equation}
\text{d}e = (Ae + G(\mu_{qp} + \varepsilon^m_i(t))\text{d}t + G\sigma_i(x)\text{d}\xi(t)
\label{eq:err_dyn_noassumption}
\end{equation}
with $e(0) = x(0) - x_{rm}(0)$ and where $\varepsilon^m_i(t)=m_i(x) - \Delta(x)$.  $i\in\mathbb{N}$ is a switching index which updates each time the model is updated.  The main problem which we address is how to find a pseudo-control $\mu_{qp}$ which provably drives the tracking error to $0$ while simultaneously guaranteeing safety.

Since $\Delta(x)$ is not known a priori, one approach is to assume that $\|\varepsilon^m_i(t)\|$ is bounded by some known term.  The size of this bound will depend on the type of model used to represent the uncertainty, its training method, and the distribution of the data $\mathcal{D}_i$.  See \cite{Chowdhary2015} for such an analysis for sparse online Gaussian Processes.  For neural networks in general there has been some work on analyzing these bounds \cite{yarotsky2017error,shi2019neural}.  
For simplicity, let us assume the modeling error $\epsilon_i^m(t)=0$, and instead rely on $\sigma_i(x)$ to fully capture any remaining modeling error in the drift.  Then we have the following dynamics:
\begin{equation}
    \label{eq:err_dyn}
    \text{d}e = (Ae + G\mu_{qp})\text{d}t + G\sigma_i(x)\text{d}\xi(t)
\end{equation}
with $e(0) = x(0) - x_{rm}(0)$.  This is valid as long as $\sigma_i(x)$ captures both the epistemic and aleatoric uncertainty accurately.  Note also that if the bounds on $\|\varepsilon_i^m(t)\|$ are known, then our results are easily extensible to this case via (\ref{eq:err_dyn_noassumption}).


\subsection{Stochastic Control Lyapunov Functions for Switched Systems}

We establish sufficient conditions on $\mu_{qp}$ to guarantee convergence of the error process $e(t)$ to 0.  The result is a linear constraint similar to deterministic CLFs (e.g., \cite{nguyen2016optimal}).  The difference here is the construction a stochastic CLF condition for switched systems.  The switching is needed to account for online updates to the model as more data is accumulated.

In general, consider a switched SDE of It\^o type \cite{khasminskii2011stochastic} defined by:
\begin{equation}
    \emph{d}X(t) = a(t,X(t))\emph{d}t+\sigma_i(t,X(t))\emph{d}\xi(t)
    \label{eq:dyn_general}
\end{equation}
where $X\in\mathbb{R}^{n_1}$, $\xi(t)\in\mathbb{R}^{n_2}$ is a Wiener process, $a(t,X)$ is a $\mathbb{R}^{n_1}$-vector function, $\sigma_i(t,X)$ a $n_1\times n_2$ matrix, and $i\in\mathbb{N}$ is a switching index.  The switching index may change a finite number of times in any finite time interval.  For each switching index, $a$ and $\sigma$ must satisfy the Lipschitz condition $\|a(t,x) - a(t,y)\| + \|\sigma_i(t,x) - \sigma_i(t,y)\| \leq L \|x-y\|, \forall x,y \in D$ with $D$ compact.  Then the solution of (\ref{eq:dyn_general}) is a continuous Markov process.

\begin{definition}
$X(t)$ is said to be \emph{exponentially mean square ultimately bounded uniformly in} i if there exists positive constants $K,c_0,\tau$ such that for all $t, X_0, i$, we have that $\mathbb{E}_{X_0}\|X(t)\|^2\leq K + c_0\|X_0\|^2 e^{-\tau t}$.
\end{definition}
We first restate the following theorem from \cite{Chowdhary2015}:
\begin{theorem}
\label{thm:1}
Let $X(t)$ be the process defined by the solution to (\ref{eq:dyn_general}), and let $V(t,X)$ be a function of class $\mathcal{C}^2$ with respect to $X$, and class $\mathcal{C}^1$ with respect to $t$.  Denote the It\^o differential generator by $\mathcal{L}$.  If 1) $-\alpha_1 + c_1\|X\|^2 \leq V(t,X)\leq c_3\|X\|^2+\alpha_2$ for real $\alpha_1,\alpha_2,c_1>0$; and 2) $\mathcal{L}V(t,X)\leq \beta_i - c_2V(t,X)$ for real $\beta_i,c_2>0$, and all i; then the process $X(t)$ is exponentially mean square ultimately bounded uniformly in $i$.  Moreover, $K=\frac{\alpha_2}{c_1} + \max_i(\frac{|\beta_i|}{c_1c_2}+\frac{\alpha_1}{c_1})$, $c_0=\frac{c_3}{c_1}$, and $\tau=c_2$.
\end{theorem}
\begin{proof}
See \cite{Chowdhary2015} Theorem 1. 
\end{proof}

We use Theorem \ref{thm:1} to derive a stochastic CLF sufficient condition on $\mu_{qp}$ for the tracking error $e(t)$.  Consider the stochastic Lyapunov candidate function $V(e) = \frac{1}{2}e^\intercal P e $
where $P$ is the solution to the Lyapunov equation $A^\intercal P + P A = -Q$, where $Q$ is any symmetric positive-definite matrix.

\begin{theorem}
\label{thm:2}
Let $e(t)$ be the switched stochastic process defined by (\ref{eq:err_dyn}), and let $\epsilon>0$ be a positive constant.  Suppose for all $t$, $\mu_{qp}$ and the relaxation variable $d_i^1\in\mathbb{R}$ satisfy the inequality:
\begin{align}
    \label{eq:clf}
    &\qquad\qquad\qquad \Psi^0_i + \Psi^1\mu_{qp}\leq d_i^1\\
    &\Psi^0_i = -\frac{1}{2}e^\intercal Q e + \frac{1}{\epsilon}V(e) + \frac{1}{2}\emph{tr}(G\sigma_i\sigma_i^\intercal G^\intercal P)\nonumber\\
    &\Psi^1 = e^\intercal P G.\nonumber
\end{align}
Then $e(t)$ is exponentially mean-square ultimately bounded uniformly in $i$.  Moreover if (\ref{eq:clf}) is satisfied with $d_i^1<0$ for all $i$, then $e(t)\rightarrow 0$ exponentially in the mean-squared sense.
\end{theorem}
\begin{proof}
The Lyapunov candidate function $V(e)$ is bounded above and below by $\frac{1}{2}\lambda_{min}(P)\|e\|^2 \leq V(e(t))\leq \frac{1}{2}\lambda_{max}(P)\|e\|^2$.  We have the following It\^o differential of the Lyapunov candidate:
\begin{align}
    \mathcal{L}V(e)&=\sum_{j}\frac{\partial V(e)}{\partial e_j}Ae_j + \frac{1}{2}\sum_{j,k}[G\sigma_i\sigma_i^\intercal G^\intercal]_{jk}\frac{\partial^2V(e)}{\partial e_k\partial e_j}\nonumber\\
    &= -\frac{1}{2}e^\intercal Q e + e^\intercal P G \mu_{qp} + \frac{1}{2}\emph{tr}(G\sigma_i\sigma_i^\intercal G^\intercal P).
\end{align}
Rearranging, (\ref{eq:clf}) becomes $\mathcal{L}V(e)\leq -\frac{1}{\epsilon}V(e)$.  Setting $\alpha_1=\alpha_2=0,\beta_i=d_i^1,c_1=\frac{1}{2}\lambda_{min}(P),c_2=\frac{1}{\epsilon},c_3=\frac{1}{2}\lambda_{max}(P)$, we see that the conditions for Theorem \ref{thm:1} are satisfied and $e(t)$ is exponentially mean square ultimately bounded uniformly in $i$.  Moreover,
\begin{multline}
    \mathbb{E}_{e_0}\|e(t)\|^2 \leq\\
    \kappa(P)\|e_0\|^2e^{-\frac{t}{\epsilon}} + \max_i(\frac{|d_i^1|}{4\lambda_{min}(P)\lambda_{max}(P)})
\end{multline}
where $\kappa(P)$ is the condition number of the matrix $P$.
Therefore if $d_i^1<0$ for all $i$, $e(t)$ converges to 0 exponentially in the mean square sense.
\end{proof}

The relaxation variable $d_i^1$ allows us to find solutions for $\mu_{qp}$ which may not always strictly satisfy a Lyapunov stability criterion $\mathcal{L}V\leq 0$.  This allows us to incorporate additional constraints on $\mu_{qp}$ at the cost of losing convergence of the error $e$ to 0.  Fortunately, the error will remain bounded by the largest $d_i^1$.  In practice we re-optimize for a new $d_i^1$ at each timestep.  This does not affect the result of Theorem \ref{thm:2} as long as we re-optimize a finite number of times for any given finite interval.  

One highly relevant set of constraints we want to satisfy are control constraints $Hu\leq b$, where $H\in\mathbb{R}^{n_c}\times\mathbb{R}^{n}$ is a matrix and $b\in\mathbb{R}^{n_c}$ is a vector.  Let $\mu_d=\mu_{rm}+\mu_{pd}-\mu_{ad}$.  Recall the pre-control law (\ref{eq:precontrol}). Then the control constraint is:
\begin{equation}
    Hg^{-1}(x)\mu_{qp}\leq H\hat{g}^{-1}(x)(\mu_d - \hat{f}(x)) + b.
\end{equation}
Next we formulate additional constraints to guarantee safety.

\subsection{Stochastic Control Barrier Functions for Switched Systems}
We leverage recent results on stochastic control barrier functions \cite{clark2019} to derive constraints linear in $\mu_{qp}$ which guarantee the process $x(t)$ satisfies a safety constraint, i.e., $x(t)\in\mathcal{C}$ for all $t$.  The set $\mathcal{C}$ is defined by a locally Lipschitz function $h: \mathbb{R}^n\rightarrow\mathbb{R}$ as $\mathcal{C}=\{x:h(x)\geq 0\}$ and $\partial\mathcal{C}=\{x:h(x)=0\}$.  We first extend the results of \cite{clark2019} to switched stochastic systems.
\begin{definition}
Let $X(t)$ be a switched stochastic process defined by (\ref{eq:dyn_general}).  Let the function $B:\mathbb{R}^n\rightarrow\mathbb{R}$ be locally Lipschitz and twice-differentiable on $\emph{int}(\mathcal{C})$.  If there exists class-K functions $\gamma_1$ and $\gamma_2$ such that for all $X$, $1/\gamma_1(h(X))\leq B(X)\leq 1/\gamma_2(h(X))$, then $B(x)$ is called a \emph{candidate control barrier function}.
\end{definition}
\begin{definition}
Let $B(x)$ be a candidate control barrier function.  If there exists a class-K function $\gamma_3$ such that $\mathcal{L}B(X)\leq \gamma_3(h(X))$, then $B(x)$ is called a \emph{control barrier function (CBF)}.
\end{definition}

\begin{theorem}
\label{thm:3}
Suppose there exists a CBF for the switched stochastic process $X(t)$ defined by (\ref{eq:dyn_general}).  If $X_0\in\mathcal{C}$, then for all $t$, $Pr(X(t)\in\mathcal{C})=1$.
\end{theorem}
\begin{proof}
\cite{clark2019} Theorem 1 provides a proof of the result for non-switched stochastic processes.  Let $t_i$ denote the switching times of $X(t)$, i.e., when $t\in[0,t_0)$, the process $X(t)$ has diffusion matrix $\sigma_0(X)$, and when $t\in[t_{i-1},t_i)$ for $i>0$, the process $X(t)$ has diffusion matrix $\sigma_i(X)$.  If $X_0\in\mathcal{C}$, then $X_{t}\in\mathcal{C}$ for all $t\in[0,t_0)$ with probability 1 since the process $X(t)$ does not switch in the time interval $t\in[0,t_0)$.  By similar argument for any $i>0$ if $X_{t_{i-1}}\in\mathcal{C}$ then $X_t\in\mathcal{C}$ for all $t\in[t_{i-1},t_i)$ with probability 1.  This also implies that  $X_{t_i}\in\mathcal{C}$, since $X(t)$ is a continuous Markov process.  Then $X_t\in\mathcal{C}$ for all $t\in[t_{i},t_{i+1})$ with probability 1.  Then by induction, for all $t$, $Pr(X(t)\in\mathcal{C})=1$.
\end{proof}
Next, we establish a linear constraint condition sufficient for $\mu_{qp}$ to guarantee safety for (\ref{eq:err_dyn}).  Rewrite (\ref{eq:err_dyn}) in terms of $x(t)$ as:
\begin{align}
    \text{d}x=(A_0x + G(\mu_d + \mu_{qp}))\text{d}t  + G\sigma_i(x)\text{d}\xi(t)\\
    A_0=\begin{bmatrix}0 & I \\ 0 & 0\end{bmatrix}, \quad \mu_d = \mu_{rm} + \mu_{pd} - \mu_{ad}.\nonumber
    \label{eq:dyn_x}
\end{align}

\begin{theorem}
\label{thm:4}
Let $x(t)$ be a switched stochastic process defined by (\ref{eq:dyn_x}).  Let $B(x)$ be a candidate control barrier function.  Let $\gamma_3$ be a class-K function.  Suppose for all $t$, $\mu_{qp}$ satisfies the inequality: 
\begin{align}
    & \qquad \qquad \qquad \qquad \Phi^0_i + \Phi^1\mu_{qp}\leq 0\\
    &\Phi^0_i = \frac{\partial B}{\partial x}^{\intercal}(A_0x + G\mu_d) - \gamma_3(h(x)) + \frac{1}{2}\emph{tr}(G\sigma_i\sigma_i^\intercal G^\intercal \frac{\partial^2 B}{\partial x^2}) \nonumber\\
    &\Phi^1 = \frac{\partial B}{\partial x}^{\intercal}G \nonumber
    \label{eq:cbf}
\end{align}
Then $B(x)$ is a CBF and (\ref{eq:cbf}) is a sufficient condition for safety, i.e., if $x_0\in\mathcal{C}$, then $x(t)\in\mathcal{C}$ for all $t$ with probability 1.
\end{theorem}

\begin{proof}
We have the following It\^o differential of the CBF candidate $B(x)$:
\begin{multline}
    \mathcal{L}B(x) = \frac{\partial B}{\partial x}^{\intercal} (A_0x + G(\mu_d + \mu_{qp}))\\
    + \frac{1}{2}\emph{tr}(G\sigma_i\sigma_i^\intercal G^\intercal \frac{\partial^2 B}{\partial x^2}).
\end{multline}
Rearranging (\ref{eq:cbf}) it is clear that $\mathcal{L}B(x)\leq \gamma_3(h(x))$.  Then $B(x)$ is a CBF and the result follows from Theorem \ref{thm:3}.
\end{proof}


\subsection{Safety and Stability under Model Adaptation}

We can now construct a CLF-CBF Quadratic Program (QP) in terms of $\mu_{qp}$ incorporating both the adaptive stochastic CLF and CBF conditions, along with control limits (Equation (\ref{eq:clf_cbf_qp})):
\begin{align}
\label{eq:clf_cbf_qp}
    \arg\min_{\mu_{qp},d_1,d_2} \quad & \mu_{qp}^\intercal \mu_{qp} + p_1d_1^2 + p_2d_2^2\\
    s.t. \quad & \Psi^0_i + \Psi^1\mu_{qp} \leq d_1 \tag{\textbf{Adaptive CLF}}\\
    \quad & \Phi^0_i + \Phi^1\mu_{qp}\leq d_2 \tag{\textbf{Adaptive CBF}}\\
    & Hg^{-1}(x)\mu_{qp}\leq Hg^{-1}(x)(\mu_d - \hat{f}(x)) + b\nonumber
\end{align}

In practice, several modifications to this QP are often made (\cite{Nguyen2016},\cite{ames}).  In addition to a relaxation term for the CLF in Theorem \ref{thm:2}, we also include a relaxation term $d^2$ for the CBF.  This helps to ensure the QP is feasible and allows for slowing down as much as possible when the safety constraint cannot be avoided due to control constraints, creating, e.g., lower impact collisions.  Safety is still guaranteed as long as the relaxation term is less than 0.  For an example of guaranteed safety in the presence of this relaxation term see \cite{nguyen2016optimal}, also see \cite{gurriet2018towards} for an approach to handling safety with control constraints.  The emphasis of this work is on guaranteeing safety in the presence of adaptation so we leave these considerations for future work.  Our entire framework is outlined in Algorithm \ref{alg:adaptiveclbf}.

\begin{algorithm}
\caption{BAyesian Learning-based Safety and Adaptation (BALSA)}
\label{alg:adaptiveclbf}
\textbf{Require:} Prior model $\hat{f}(x)$, known $g(x)$, reference trajectory $x_{rm}$, choice of modeling algorithm $\bar{\Delta}_i(x)\sim\mathcal{N}(m_i(x),\sigma_i(x))$, $dt$, $A$, $Hu\leq b$. \\
\textbf{Initialize:} $i=0$, Dataset $\mathcal{D}_0=\emptyset$, $t=0$, solve $P$\\
\While{true}{
Obtain $\mu_{rm}=\dot{x}_{2rm}(t)$ and compute $\mu_{pd}$\\
Compute model error and uncertainty $\mu_{ad}=m_i(x(t))$, and $\sigma_i(x(t))$ \\
$\mu_{qp} \leftarrow $ Solve QP (\ref{eq:clf_cbf_qp})\\
Set $u(t) = g(x)^{-1}(\mu_{rm} + \mu_{pd} + \mu_{qp} -\mu_{ad} - \hat{f}(x))$\\
Apply control $u(t)$ to system.\\
Step forward in time $t\leftarrow t+dt$.\\
Append new data point to database:\\
$\bar{X}_t = [x(t)]$, $\bar{Y}_t=(x_2(t+dt) - x_2(t))/dt - (\hat{f}(x(t)) + g(x(t)u(t))$. \\
$\mathcal{D}_{i}\leftarrow\mathcal{D}_i\cup\{\bar{X}_t,\bar{Y}_t\}$\\
\If{updateModel}{
    Update model $\bar{\Delta}_i(x,\mu)$ with database $\mathcal{D}_i$\\
    $\mathcal{D}_{i+1}\leftarrow\mathcal{D}_i$, $i\leftarrow i+1$
    }
}
\end{algorithm}


\section{Application to Fast Autonomous Driving}

In this section we validate BALSA on a kinematic bicycle model for car-like vehicles.  We model the state $x=[p_x,p_y,\theta,v]^\intercal$ as position in x and y, heading, and velocity respectively, with dynamics $\dot{x} = [v\cos(\theta),v\sin(\theta),v\tan(\psi)/L,a]^\intercal$.
where $a$ is the input acceleration, $L$ is the vehicle length, and $\psi$ is the steering angle.  We employ a simple transformation to obtain dynamics in the form of (\ref{eq:dynamics}).  Let $z=[z_1,z_2,z_3,z_4]^\intercal$ where $z_1=p_x$, $z_2=p_y$, $z_3=\dot{z}_1$, $z_4=\dot{z}_2$, and $c=\tan(\psi)/L$.  Let the controls $u=[c, a]^\intercal$.  Then $\dot{z}$ fits the canonical form of (\ref{eq:dynamics}).
To ascertain the importance of learning and adaptation, we add the following disturbance to $[\dot{z}_3, \dot{z}_4]^\intercal$ to use as a ``true" model:
\begin{equation}
    \delta(z) = \begin{bmatrix}\cos(\theta) & -\sin(\theta)\\\sin(\theta) & \cos(\theta)\end{bmatrix} \begin{bmatrix}-\tanh(v^2)\\-(0.1 + v)\end{bmatrix}
    \label{eq:car_dyn_true}
\end{equation}

This constitutes a non-linearity in the forward velocity and a tendency to drift to the right.

We use the following barrier function for pointcloud-based obstacles.  Similar to \cite{nguyen2016optimal}, we design this barrier function with an extra component to account for position-based constraints which have a relative degree greater than 1.  This is done by including the time-derivative of the position-based constraint as an additional term in the barrier function, which penalizes velocities (or higher order derivatives) leading to a decrease of the level set function $h$.  Let our safety set $\mathcal{C}=\{x\in\mathbb{R}^n|h(x,x')\geq0\}$, where $x'$ is the position of an obstacle.  Let $h(x,x') = \|(x-x')\|_2-r$ where $r>0$ is the radius of a circle around the obstacle.  Then construct a barrier function $B(x;x')= 1/(\gamma_p h(x,x')+\frac{d}{dt}h(x,x'))$.
As shown by \cite{Nguyen2016}, $B(x)$ is a CBF, where $\gamma_p$ helps to control the rate of convergence. We chose $\gamma_1(x),\gamma_2(x)=x$ and $\gamma_3(x)=\gamma/x$.

\subsection{Validation of BALSA in Simulation}

One iteration of the algorithm for this problem takes less than $4ms$ on a 3.7GHz Intel Core i7-8700K CPU, in Python code which has not been optimized for speed.  We make our code publicly available\footnote{\href{https://github.com/ddfan/balsa.git}{https://github.com/ddfan/balsa.git}}.  Because training the model occurs on a separate thread and can be performed anytime online, we do not include the model training time in this benchmark.  We use OSQP \cite{osqp} as our QP solver.

In Figure \ref{fig:toy_barriers}, we compare BALSA with several different baseline algorithms.  We use a Neural Network trained with dropout and a negative-log-likelihood loss function for capturing the uncertainty \cite{gal2016dropout}.  We place several obstacles in the direct path of the reference trajectory.  We also place velocity barriers for driving too fast or too slow.  We observe that the behavior of the vehicle using our algorithm maintains good tracking errors while avoiding barriers and maintaining safety, while the other approaches suffer from various drawbacks.  The adaptive controller (ad) and PD controller (pd) violate safety constraints.  The (qp) controller with an inaccurate model also violates constraints and exhibits highly suboptimal behavior (Figure \ref{fig:errors_losses}).  A robust (rob) formulation which uses a fixed robust bound which is meant to bound any model uncertainty \cite{nguyen2016optimal}, while not violating safety constraints, is too conservative and non-adaptive, has trouble tracking the reference trajectory.  In contrast, BALSA adapts to model error with guaranteed safety.  We also plot the model uncertainty and error in (Figure \ref{fig:errors_losses}).

\begin{figure}
    \centering
    \includegraphics[width=\linewidth]{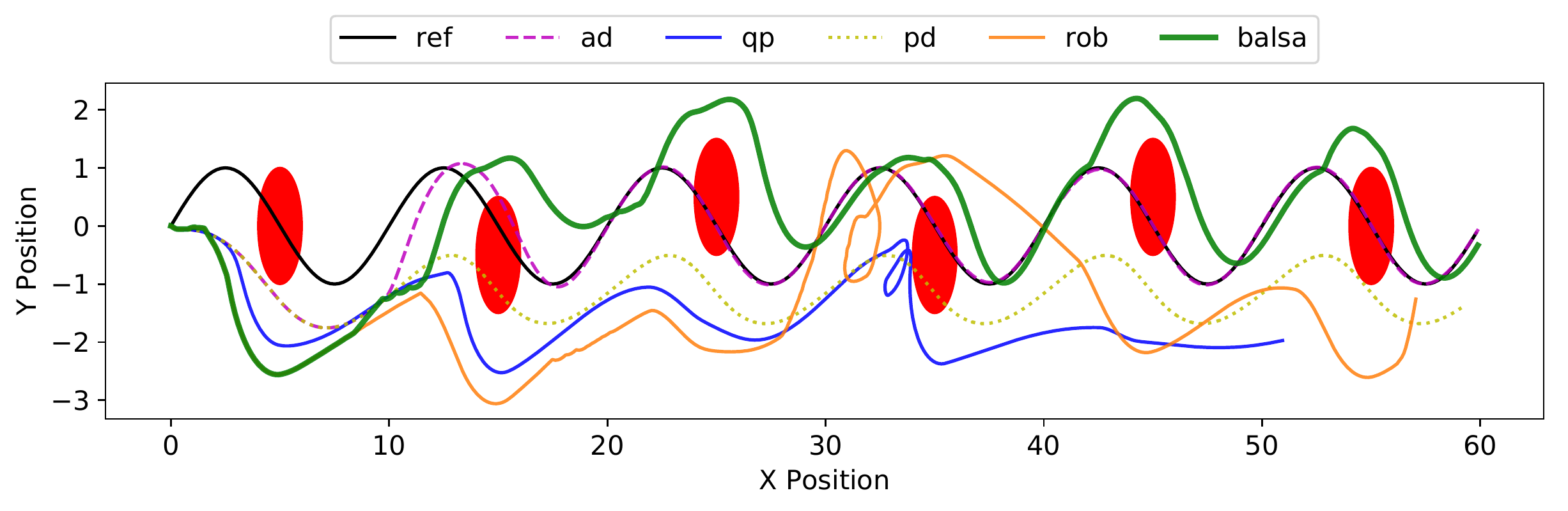}
    \caption{Comparison of the performance of four algorithms in tracking and avoiding barrier regions (red ovals). \textbf{ref} is the reference trajectory.  \textbf{ad} is an adaptive controller ($\mu_{rm} + \mu_{pd} - \mu_{ad})$.  \textbf{qp} is a non-adaptive safety controller ($\mu_{rm} + \mu_{pd} + \mu_{qp}$).  \textbf{pd} is a proportional derivative controller ($\mu_{rm}$ + $\mu_{pd}$).  \textbf{rob} is a robust controller which uses a fixed $\sigma_i(x)$ to compensate for modeling errors.  \textbf{balsa} is the full adaptive CLF-CBF-QP approach outlined in this paper and in Algorithm \ref{alg:adaptiveclbf}, i.e. ($\mu_{rm} + \mu_{pd} - \mu_{ad} + \mu_{qp}$).
    }
    \label{fig:toy_barriers}
\end{figure}

\begin{figure}
  \centering
  \includegraphics[width=\linewidth]{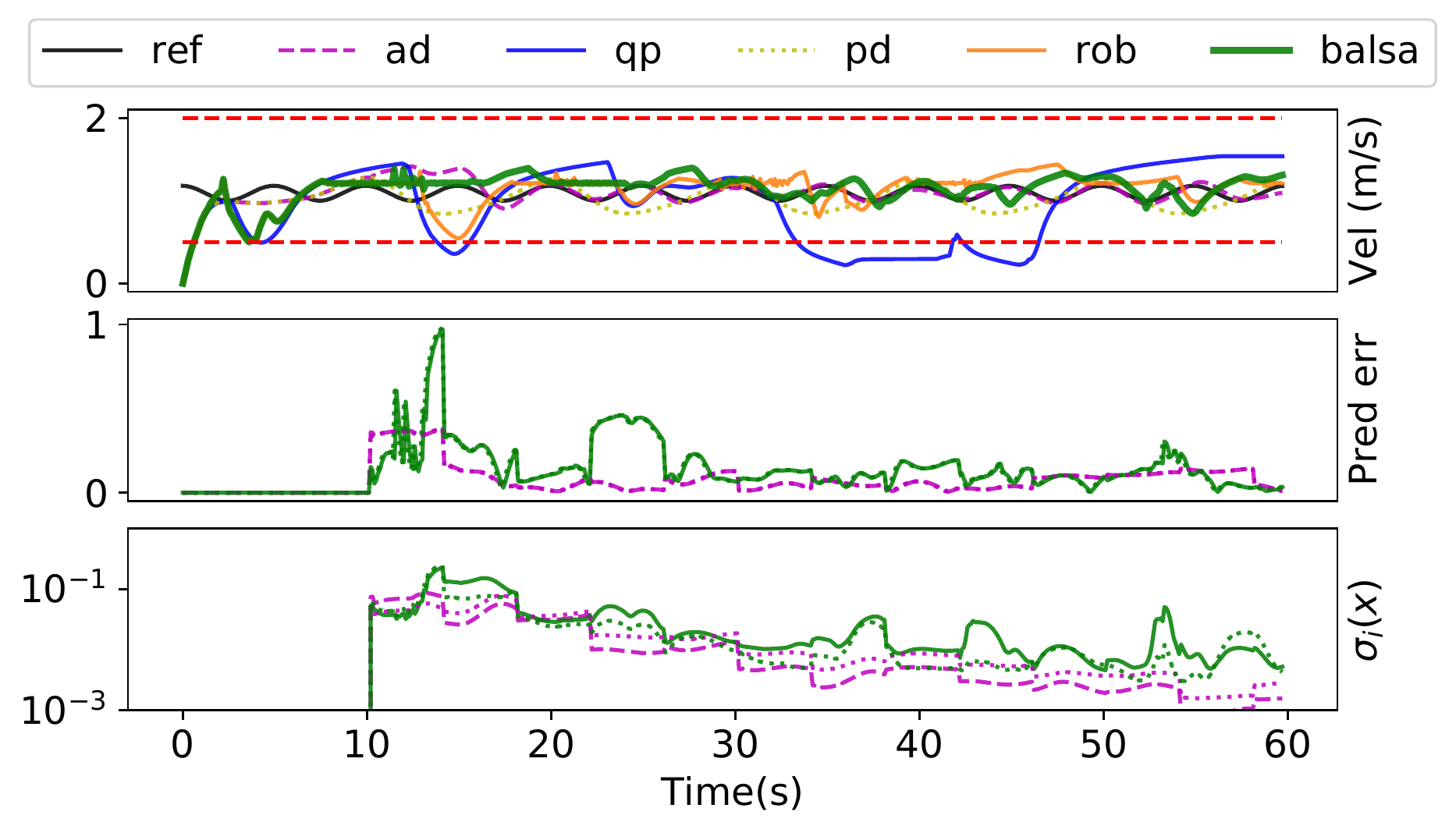}
  \caption{Top: Velocities of each algorithm.  Red dotted line indicates safety barrier. Middle: Output prediction error of model, decreasing with time.  Solid and dashed lines indicate both output dimensions. Bottom: Uncertainty $\sigma_i(x)$, also decreasing with time.  Predictions are made after 10 seconds to accumulate enough data to train the network.  During this time we choose an upper bound for $\sigma_0=1.0$.}
  \label{fig:errors_losses}
\end{figure}

\subsection{Comparing Different Modeling Methods in Simulation}

Next we compared the performance of BALSA on three different Bayesian modeling algorithms:  Gaussian Processes, a Neural Network with dropout, and ALPaCA \cite{Harrison2018}, a meta-learning approach which uses a hybrid neural network with Bayesian regression on the last layer. For all methods we retrained the model intermittently, every 40 new datapoints.  In addition to the current state, we also included as input to the model the previous control, angular velocity in yaw, and the current roll and pitch of the vehicle.  For the GP we re-optimized hyperparameters with each training.  For the dropout NN, we used 4 fully-connected layers with 256 hidden units each, and trained for 50 epochs with a batch size of 64. Lastly, for ALPaCA we used 2 hidden layers, each with 128 units, and 128 basis functions.  We used a batch size of 150, 20 context data points, and 20 test data points. The model was trained using 100 gradient steps and online adaption (during prediction) was performed using 20 of the most recent context data points with the current observation (see \cite{Harrison2018} for details of the meta-learning capabilities of ALPaCA).  At each training iteration we retrain both the neural network and the last Bayesian linear regression layer.  Figure (\ref{fig:model_compare}) and Table (\ref{table:1}) show a comparison of tracking error for these methods.  We found GPs to be computationally intractable with more than 500 data points, although they exhibited good performance.  Neural networks with dropout converged quickly and were efficient to train and run.  ALPaCA exhibited slightly slower convergence but good tracking as well.

\begin{figure}
\includegraphics[width=0.35\linewidth]{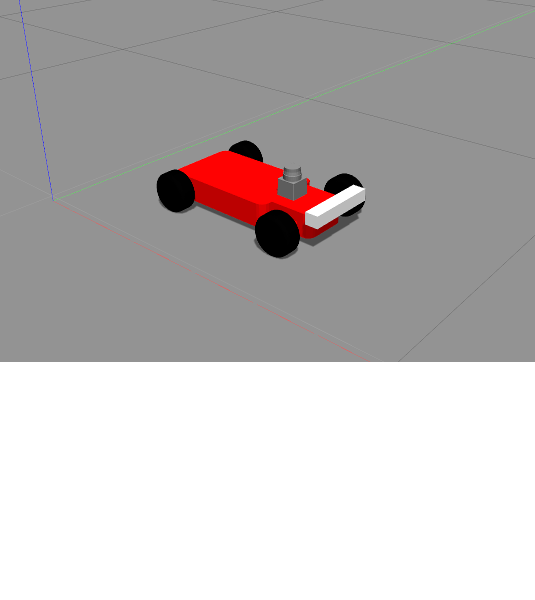}
\includegraphics[width=0.6\linewidth]{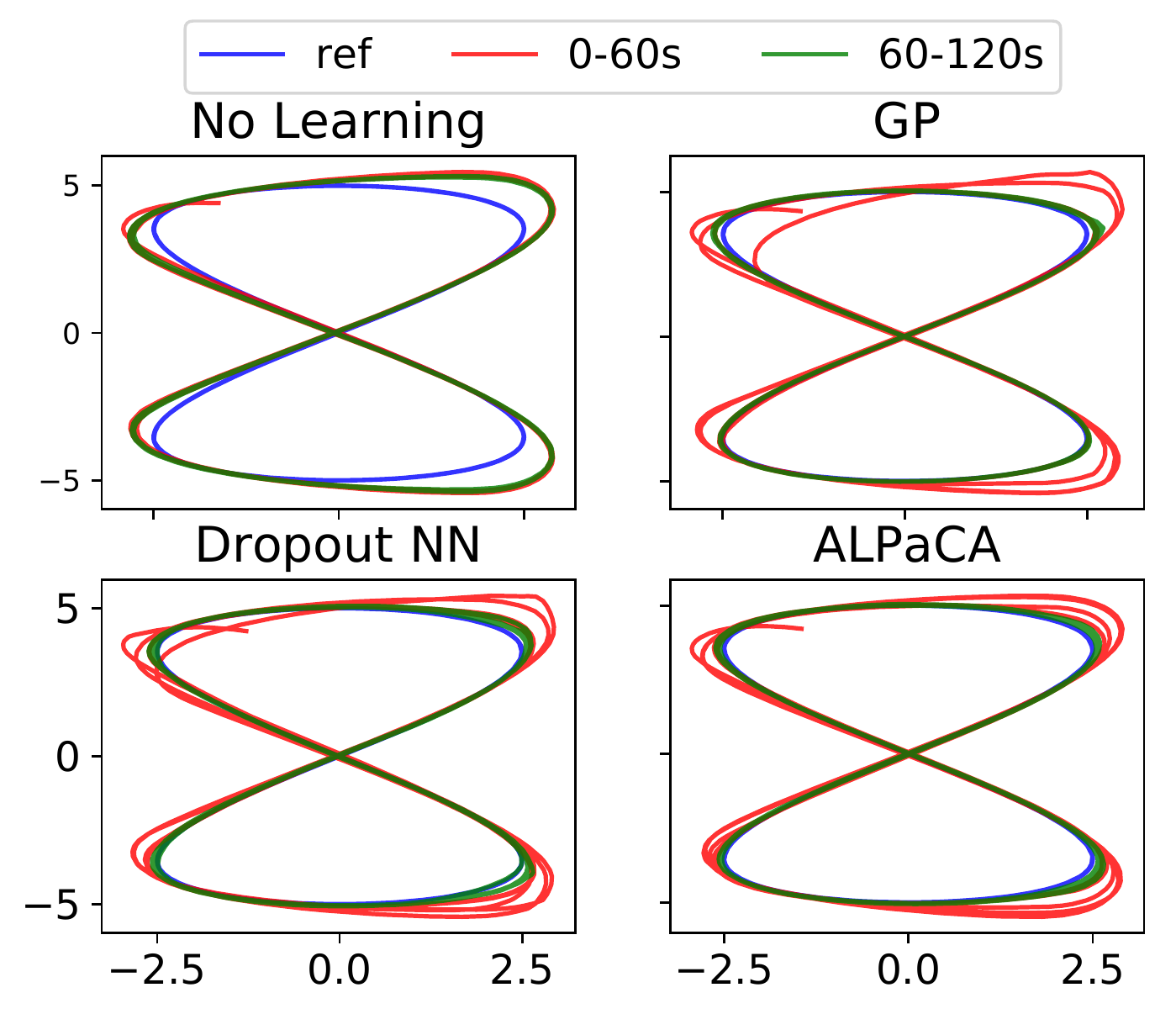}
\caption{Comparison of adaptation performance in a Gazebo simulation using three different probabilistic model learning methods.}
\label{fig:model_compare}
\end{figure}

\begin{table}
\centering
\begin{tabular}{ |c|c|c|c|c| } 
\hline
 & No learn & GP & Dropout & ALPaCA \\
\hline
\hline
0-60s & 0.580 & 0.3992 & 0.408 & 0.390 \\
\hline
60-120s & 0.522 & 0.097 & 0.105 & 0.110 \\ 
\hline
\end{tabular}
\caption{Average tracking error in position for different modeling methods in sim, split into the first minute and second minute.}
\label{table:1}
\end{table}

\subsection{Hardware Experiments on Martian Terrain}
To validate that BALSA meets real-time computational requirements, we conducted  hardware experiments on the platform depicted in Figure (\ref{fig:adapt_compare_hw}).  We used an off-the shelf RC car (Traxxas Xmaxx) in 1/5-th scale (wheelbase 0.48 m), equipped with sensors such as a 3D LiDAR (Velodyne VLP-16) for obstacle avoidance and a stereo camera (RealSense T265) for on-board for state estimation.  The power train consists of a single brushless DC motor, which drives the front and rear differential, operating in current control mode for controlling acceleration.  Steering commands were fed to a servo position controller.  The on-board computer (Intel NUC i7) ran Ubuntu 18.04 and ROS \cite{quigley2009ros}.

Experiments were conducted in a Martian simulation environment, which contains sandy soil, gravel, rocks, and rough terrain.  We gave figure-eight reference trajectories at 2m/s and evaluated the vehicle's tracking performance (Figure \ref{fig:adapt_compare_hw}).  Due to large achieving good tracking performance at higher speeds is difficult.  We observed that BALSA is able to adapt to bumps and changes in friction, wheel slip, etc., exhibiting improved tracking performance over a non-adaptive baseline (Table \ref{table:2}).

\begin{figure}
\includegraphics[trim=20 20 0 15,clip,width=0.35\linewidth]{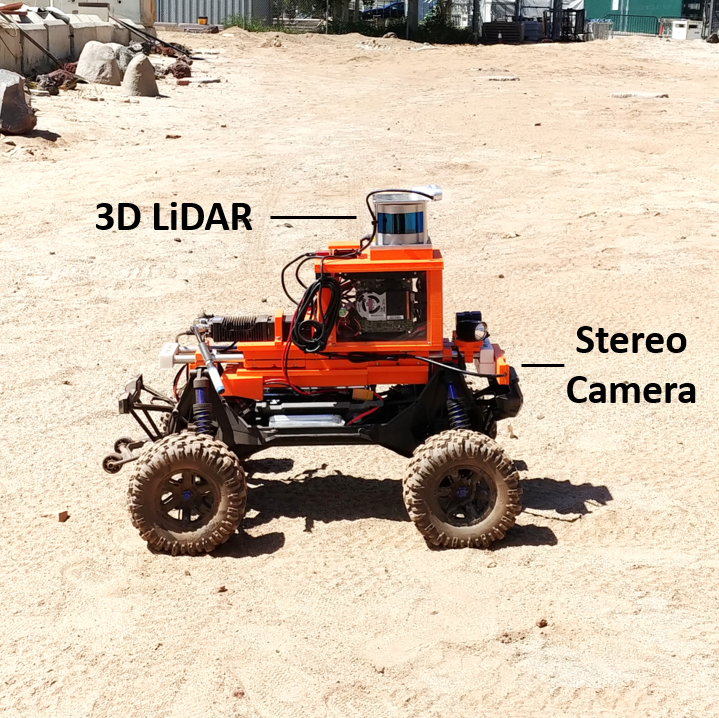}
\includegraphics[width=0.60\linewidth]{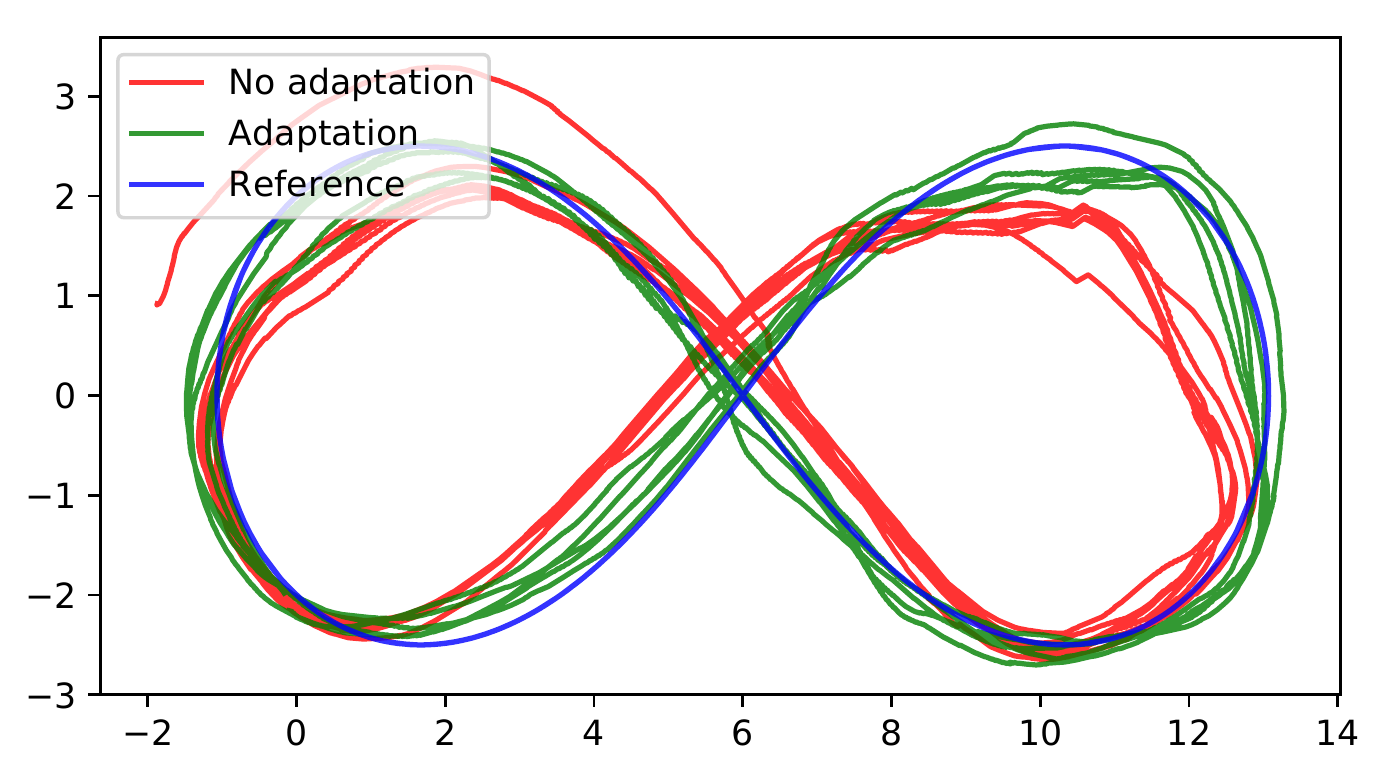}
\caption{Left: A high-speed rover vehicle.  Right: Figure-8 tracking on our rover platform on rough and sandy terrain, comparing adaptation vs. no adaptation.}
\label{fig:adapt_compare_hw}
\end{figure}

\begin{table}
\centering
\begin{tabular}{ |c|c|c|c| } 
\hline
 & Mean Err & Std Dev & Max \\
\hline
\hline
No Learn & 1.417 & 0.568 & 6.003 \\
\hline
Learning & 0.799 & 0.387 & 2.310 \\ 
\hline
\end{tabular}
\caption{Mean, standard deviation, and max tracking error on our rover platform for a figure-8 task.}
\label{table:2}
\end{table}

We also evaluated the safety of BALSA under adaptation.  We used LiDAR pointclouds to create barriers at each LiDAR return location.  Although this creates a large number of constraints, the QP solver is able to handle these in real-time.  Figure \ref{fig:collision_avoid_hw} shows what happens when an obstacle is placed in the path of the reference trajectory.  The vehicle successfully slows down and comes to a stop if needed, avoiding the obstacle altogether.

\begin{figure}
  \centering   
  \begin{overpic}[width=0.8\linewidth]{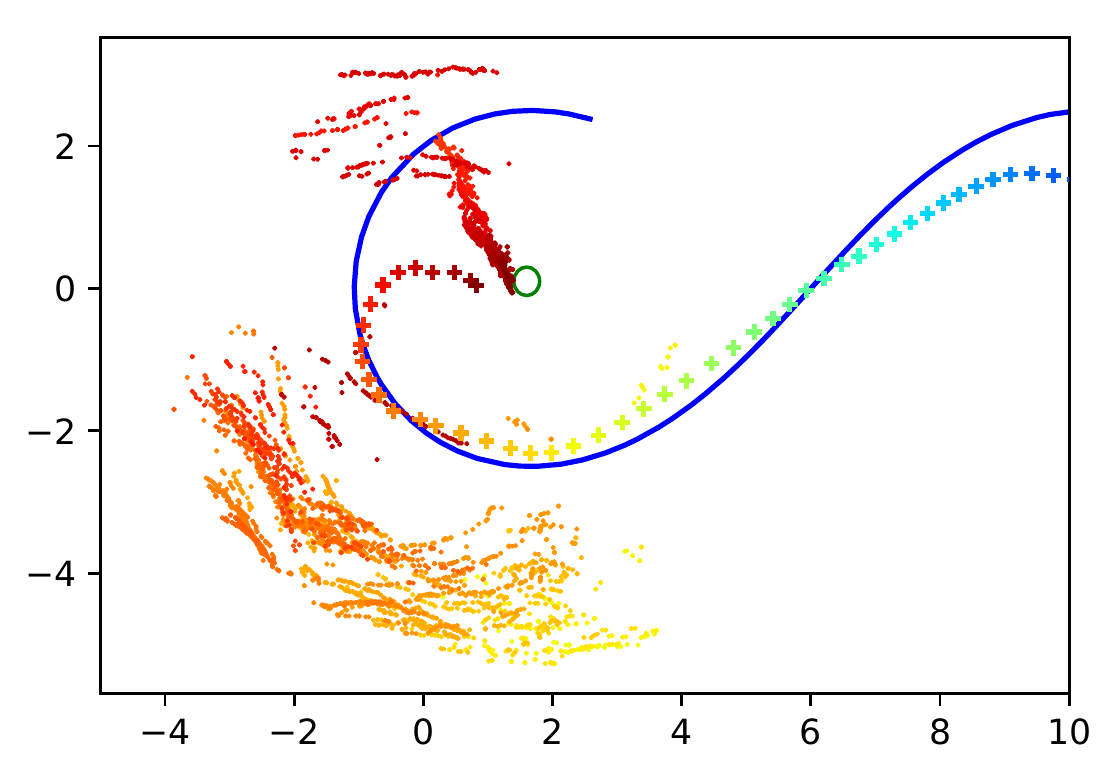}
     \put(58,9){\includegraphics[scale=0.24]{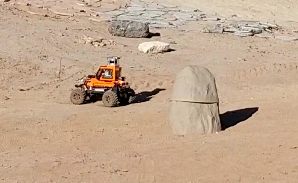}}  
  \end{overpic}
\caption{Vehicle avoids collision despite localization drift and unmodeled dynamics.  Blue line is the reference trajectory, colored pluses are the vehicle pose, colored points are obstacles.  Colors indicate time, from blue (earlier) to red (later).  Note that localization drift results in the obstacles appearing to shift position.  Green circle indicates location of the obstacle at the last timestep.  Despite this drift the vehicle does not collide with the obstacle.}
\label{fig:collision_avoid_hw}
\end{figure}


\section{Conclusion}
     
In this work, we have described a framework for safe, fast, and computationally efficient probabilistic learning-based control.  The proposed approach satisfies several important real-world requirements and take steps towards enabling safe deployment of high-dimensional data-driven controls and planning algorithms.  Further development other types of robots including drones, legged robots, and manipulators is straightforward.  Incorporating better uncertainty-representing modeling methods and training on higher-dimensional data (vision, LiDAR, etc) will also be a fruitful direction of research.


\section*{Acknowledgement}

The authors would like to thank Joel Burdick's group for their hardware support.  This research was partially carried out at the Jet Propulsion Laboratory (JPL), California Institute of Technology, and was sponsored by the JPL Year Round Internship Program and the National Aeronautics and Space Administration (NASA).  Jennifer Nguyen was supported in part by NASA EPSCoR Research Cooperative Agreement WV-80NSSC17M0053 and NASA
West Virginia Space Grant Consortium, Training Grant \#NX15AI01H.  Evangelos A. Theodorou was supported by the C-STAR Faculty Fellowship at Georgia Institute of Technology.  Copyright \textcopyright 2019. All rights reserved.

\bibliographystyle{IEEEtran}
\bibliography{main}

\begin{thebibliography}{10}
\providecommand{\url}[1]{#1}
\csname url@samestyle\endcsname
\providecommand{\newblock}{\relax}
\providecommand{\bibinfo}[2]{#2}
\providecommand{\BIBentrySTDinterwordspacing}{\spaceskip=0pt\relax}
\providecommand{\BIBentryALTinterwordstretchfactor}{4}
\providecommand{\BIBentryALTinterwordspacing}{\spaceskip=\fontdimen2\font plus
\BIBentryALTinterwordstretchfactor\fontdimen3\font minus
  \fontdimen4\font\relax}
\providecommand{\BIBforeignlanguage}[2]{{%
\expandafter\ifx\csname l@#1\endcsname\relax
\typeout{** WARNING: IEEEtran.bst: No hyphenation pattern has been}%
\typeout{** loaded for the language `#1'. Using the pattern for}%
\typeout{** the default language instead.}%
\else
\language=\csname l@#1\endcsname
\fi
#2}}
\providecommand{\BIBdecl}{\relax}
\BIBdecl

\bibitem{silver2017mastering}
D.~Silver, J.~Schrittwieser, K.~Simonyan, I.~Antonoglou, A.~Huang, A.~Guez,
  T.~Hubert, L.~Baker, M.~Lai, A.~Bolton \emph{et~al.}, ``Mastering the game of
  go without human knowledge,'' \emph{Nature}, vol. 550, no. 7676, p. 354,
  2017.

\bibitem{nasa_curiosity}
\BIBentryALTinterwordspacing
NASA, ``Where is {Curiosity}? - {NASA Mars Curiosity Rover},'' 2018. [Online].
  Available: \url{https://mars.nasa.gov/msl/mission/whereistherovernow/}
\BIBentrySTDinterwordspacing

\bibitem{nasa_opportunity}
\BIBentryALTinterwordspacing
NASA, ``Opportunity {Updates},'' 2018. [Online]. Available:
  \url{https://mars.nasa.gov/mer/mission/rover-status/opportunity/recent/all/}
\BIBentrySTDinterwordspacing

\bibitem{2014_klein}
E.~{Klein}, E.~{Nilsen}, A.~{Nicholas}, C.~{Whetsel}, J.~{Parrish},
  R.~{Mattingly}, and L.~{May}, ``The mobile mav concept for mars sample
  return,'' in \emph{2014 IEEE Aerospace Conference}, March 2014, pp. 1--9.

\bibitem{2019_nelessen}
A.~{Nelessen}, C.~{Sackier}, I.~{Clark}, P.~{Brugarolas}, G.~{Villar},
  A.~{Chen}, A.~{Stehura}, R.~{Otero}, E.~{Stilley}, D.~{Way}, K.~{Edquist},
  S.~{Mohan}, C.~{Giovingo}, and M.~{Lefland}, ``Mars 2020 entry, descent, and
  landing system overview,'' in \emph{2019 IEEE Aerospace Conference}, March
  2019, pp. 1--20.

\bibitem{wagner}
\BIBentryALTinterwordspacing
N.~Wagener, C.~Cheng, J.~Sacks, and B.~Boots, ``An online learning approach to
  model predictive control,'' \emph{CoRR}, vol. abs/1902.08967, 2019. [Online].
  Available: \url{http://arxiv.org/abs/1902.08967}
\BIBentrySTDinterwordspacing

\bibitem{williams2018information}
G.~Williams, P.~Drews, B.~Goldfain, J.~M. Rehg, and E.~A. Theodorou,
  ``Information-theoretic model predictive control: Theory and applications to
  autonomous driving,'' \emph{IEEE Transactions on Robotics}, vol.~34, no.~6,
  pp. 1603--1622, 2018.

\bibitem{berkenkamp2017safe}
F.~Berkenkamp, M.~Turchetta, A.~Schoellig, and A.~Krause, ``Safe model-based
  reinforcement learning with stability guarantees,'' in \emph{Advances in
  neural information processing systems}, 2017, pp. 908--918.

\bibitem{kim2019bi}
S.-K. Kim, R.~Thakker, and A.-A. Agha-Mohammadi, ``Bi-directional value
  learning for risk-aware planning under uncertainty,'' \emph{IEEE Robotics and
  Automation Letters}, vol.~4, no.~3, pp. 2493--2500, 2019.

\bibitem{ross2011reduction}
S.~Ross, G.~Gordon, and D.~Bagnell, ``A reduction of imitation learning and
  structured prediction to no-regret online learning,'' in \emph{Proceedings of
  the fourteenth international conference on artificial intelligence and
  statistics}, 2011, pp. 627--635.

\bibitem{Ostafew2016a}
\BIBentryALTinterwordspacing
C.~J. Ostafew, A.~P. Schoellig, and T.~D. Barfoot, ``{Robust Constrained
  Learning-based NMPC enabling reliable mobile robot path tracking},''
  \emph{The International Journal of Robotics Research}, vol.~35, no.~13, pp.
  1547--1563, nov 2016. [Online]. Available:
  \url{http://journals.sagepub.com/doi/10.1177/0278364916645661}
\BIBentrySTDinterwordspacing

\bibitem{Pereida2018}
\BIBentryALTinterwordspacing
K.~Pereida and A.~P. Schoellig, ``{Adaptive Model Predictive Control for
  High-Accuracy Trajectory Tracking in Changing Conditions},'' in \emph{2018
  IEEE/RSJ International Conference on Intelligent Robots and Systems
  (IROS)}.\hskip 1em plus 0.5em minus 0.4em\relax IEEE, oct 2018, pp.
  7831--7837. [Online]. Available:
  \url{https://ieeexplore.ieee.org/document/8594267/}
\BIBentrySTDinterwordspacing

\bibitem{Hewing2017}
\BIBentryALTinterwordspacing
L.~Hewing, J.~Kabzan, and M.~N. Zeilinger, ``{Cautious Model Predictive Control
  using Gaussian Process Regression},'' \emph{arXiv}, may 2017. [Online].
  Available: \url{http://arxiv.org/abs/1705.10702}
\BIBentrySTDinterwordspacing

\bibitem{Shi2018}
\BIBentryALTinterwordspacing
G.~Shi, X.~Shi, M.~O'Connell, R.~Yu, K.~Azizzadenesheli, A.~Anandkumar, Y.~Yue,
  and S.-J. Chung, ``{Neural Lander: Stable Drone Landing Control using Learned
  Dynamics},'' \emph{arXiv}, nov 2018. [Online]. Available:
  \url{http://arxiv.org/abs/1811.08027}
\BIBentrySTDinterwordspacing

\bibitem{Chowdhary2015}
\BIBentryALTinterwordspacing
G.~Chowdhary, H.~A. Kingravi, J.~P. How, and P.~A. Vela, ``{Bayesian
  Nonparametric Adaptive Control Using Gaussian Processes},'' \emph{IEEE
  Transactions on Neural Networks and Learning Systems}, vol.~26, no.~3, pp.
  537--550, mar 2015. [Online]. Available:
  \url{http://ieeexplore.ieee.org/document/6823109/}
\BIBentrySTDinterwordspacing

\bibitem{Nguyen}
Q.~Nguyen and K.~Sreenath, ``{Optimal Robust Control for Bipedal Robots through
  Control Lyapunov Function based Quadratic Programs.}'' \emph{Robotics:
  Science and Systems}, 2015.

\bibitem{nguyen2016optimal}
Q.~Nguyen and K.~Sreenath, ``Optimal robust control for constrained nonlinear
  hybrid systems with application to bipedal locomotion,'' in \emph{2016
  American Control Conference (ACC)}.\hskip 1em plus 0.5em minus 0.4em\relax
  IEEE, 2016, pp. 4807--4813.

\bibitem{Cheng2019}
\BIBentryALTinterwordspacing
R.~Cheng, G.~Orosz, R.~M. Murray, and J.~W. Burdick, ``{End-to-End Safe
  Reinforcement Learning through Barrier Functions for Safety-Critical
  Continuous Control Tasks},'' \emph{arXiv}, mar 2019. [Online]. Available:
  \url{http://arxiv.org/abs/1903.08792}
\BIBentrySTDinterwordspacing

\bibitem{Nguyen2015}
\BIBentryALTinterwordspacing
Q.~Nguyen and K.~Sreenath, ``{L1 adaptive control for bipedal robots with
  control Lyapunov function based quadratic programs},'' in \emph{2015 American
  Control Conference (ACC)}.\hskip 1em plus 0.5em minus 0.4em\relax IEEE, jul
  2015, pp. 862--867. [Online]. Available:
  \url{http://ieeexplore.ieee.org/document/7170842/}
\BIBentrySTDinterwordspacing

\bibitem{Taylor2019}
\BIBentryALTinterwordspacing
A.~J. Taylor, V.~D. Dorobantu, M.~Krishnamoorthy, H.~M. Le, Y.~Yue, and A.~D.
  Ames, ``{A Control Lyapunov Perspective on Episodic Learning via Projection
  to State Stability},'' \emph{arXiv}, mar 2019. [Online]. Available:
  \url{http://arxiv.org/abs/1903.07214}
\BIBentrySTDinterwordspacing

\bibitem{gurriet2018towards}
T.~Gurriet, A.~Singletary, J.~Reher, L.~Ciarletta, E.~Feron, and A.~Ames,
  ``Towards a framework for realizable safety critical control through active
  set invariance,'' in \emph{Proceedings of the 9th ACM/IEEE International
  Conference on Cyber-Physical Systems}.\hskip 1em plus 0.5em minus 0.4em\relax
  IEEE Press, 2018, pp. 98--106.

\bibitem{azimi2018robust}
V.~Azimi and P.~A. Vela, ``Robust adaptive quadratic programming and safety
  performance of nonlinear systems with unstructured uncertainties,'' in
  \emph{2018 IEEE Conference on Decision and Control (CDC)}.\hskip 1em plus
  0.5em minus 0.4em\relax IEEE, 2018, pp. 5536--5543.

\bibitem{azimi2018performance}
V.~Azimi and P.~A. Vela, ``Performance reference adaptive control: A joint
  quadratic programming and adaptive control framework,'' in \emph{2018 Annual
  American Control Conference (ACC)}.\hskip 1em plus 0.5em minus 0.4em\relax
  IEEE, 2018, pp. 1827--1834.

\bibitem{Nguyen2016}
\BIBentryALTinterwordspacing
Q.~Nguyen and K.~Sreenath, ``{Exponential Control Barrier Functions for
  enforcing high relative-degree safety-critical constraints},'' in \emph{2016
  American Control Conference (ACC)}.\hskip 1em plus 0.5em minus 0.4em\relax
  IEEE, jul 2016, pp. 322--328. [Online]. Available:
  \url{http://ieeexplore.ieee.org/document/7524935/}
\BIBentrySTDinterwordspacing

\bibitem{ames2014rapidly}
A.~D. Ames, K.~Galloway, K.~Sreenath, and J.~W. Grizzle, ``Rapidly
  exponentially stabilizing control lyapunov functions and hybrid zero
  dynamics,'' \emph{IEEE Transactions on Automatic Control}, vol.~59, no.~4,
  pp. 876--891, 2014.

\bibitem{look2019differential}
A.~Look and M.~Kandemir, ``Differential bayesian neural nets,'' \emph{arXiv
  preprint arXiv:1912.00796}, 2019.

\bibitem{liu2019neural}
X.~Liu, S.~Si, Q.~Cao, S.~Kumar, and C.-J. Hsieh, ``Neural sde: Stabilizing
  neural ode networks with stochastic noise,'' \emph{arXiv preprint
  arXiv:1906.02355}, 2019.

\bibitem{hegde2019deep}
P.~Hegde, M.~Heinonen, H.~L{\"a}hdesm{\"a}ki, S.~Kaski \emph{et~al.}, ``Deep
  learning with differential gaussian process flows,'' in \emph{International
  Conference on Artificial Intelligence and Statistics}.\hskip 1em plus 0.5em
  minus 0.4em\relax PMLR, 2019.

\bibitem{li2011knows}
L.~Li, M.~L. Littman, T.~J. Walsh, and A.~L. Strehl, ``Knows what it knows: a
  framework for self-aware learning,'' \emph{Machine learning}, vol.~82, no.~3,
  pp. 399--443, 2011.

\bibitem{Roy2011}
C.~J. Roy and W.~L. Oberkampf, ``{A comprehensive framework for verification,
  validation, and uncertainty quantification in scientific computing},''
  \emph{Computer Methods in Applied Mechanics and Engineering}, vol. 200, no.
  25-28, pp. 2131--2144, jun 2011.

\bibitem{lew2020unpublished}
\BIBentryALTinterwordspacing
T.~Lew, A.~Sharma, J.~Harrison, and M.~Pavone, ``{On the Problem of
  Reformulating Systems with Uncertain Dynamics as a Stochastic Differential
  Equation.
  http://asl.stanford.edu/wp-content/papercite-data/pdf/dynsSDE.pdf},''
  \emph{Technical Report}, 2020. [Online]. Available:
  \url{http://asl.stanford.edu/wp-content/papercite-data/pdf/dynsSDE.pdf}
\BIBentrySTDinterwordspacing

\bibitem{Hafner2018}
\BIBentryALTinterwordspacing
D.~Hafner, D.~Tran, T.~Lillicrap, A.~Irpan, and J.~Davidson, ``{Reliable
  Uncertainty Estimates in Deep Neural Networks using Noise Contrastive
  Priors},'' \emph{arXiv}, jul 2018. [Online]. Available:
  \url{http://arxiv.org/abs/1807.09289}
\BIBentrySTDinterwordspacing

\bibitem{Harrison2018}
\BIBentryALTinterwordspacing
J.~Harrison, A.~Sharma, and M.~Pavone, ``{Meta-Learning Priors for Efficient
  Online Bayesian Regression},'' \emph{arXiv}, jul 2018. [Online]. Available:
  \url{http://arxiv.org/abs/1807.08912}
\BIBentrySTDinterwordspacing

\bibitem{gal2016dropout}
Y.~Gal and Z.~Ghahramani, ``Dropout as a bayesian approximation: Representing
  model uncertainty in deep learning,'' in \emph{international conference on
  machine learning}, 2016, pp. 1050--1059.

\bibitem{shahriari2015taking}
B.~Shahriari, K.~Swersky, Z.~Wang, R.~P. Adams, and N.~De~Freitas, ``Taking the
  human out of the loop: A review of bayesian optimization,'' \emph{Proceedings
  of the IEEE}, vol. 104, no.~1, pp. 148--175, 2015.

\bibitem{pan2017prediction}
Y.~Pan, X.~Yan, E.~A. Theodorou, and B.~Boots, ``Prediction under uncertainty
  in sparse spectrum gaussian processes with applications to filtering and
  control,'' in \emph{Proceedings of the 34th International Conference on
  Machine Learning-Volume 70}.\hskip 1em plus 0.5em minus 0.4em\relax JMLR.
  org, 2017, pp. 2760--2768.

\bibitem{yarotsky2017error}
D.~Yarotsky, ``Error bounds for approximations with deep relu networks,''
  \emph{Neural Networks}, vol.~94, pp. 103--114, 2017.

\bibitem{shi2019neural}
G.~Shi, X.~Shi, M.~O’Connell, R.~Yu, K.~Azizzadenesheli, A.~Anandkumar,
  Y.~Yue, and S.-J. Chung, ``Neural lander: Stable drone landing control using
  learned dynamics,'' in \emph{2019 International Conference on Robotics and
  Automation (ICRA)}.\hskip 1em plus 0.5em minus 0.4em\relax IEEE, 2019, pp.
  9784--9790.

\bibitem{khasminskii2011stochastic}
R.~Khasminskii, \emph{Stochastic stability of differential equations}.\hskip
  1em plus 0.5em minus 0.4em\relax Springer Science \& Business Media, 2011,
  vol.~66.

\bibitem{clark2019}
A.~{Clark}, ``Control barrier functions for complete and incomplete information
  stochastic systems,'' in \emph{2019 American Control Conference (ACC)}, July
  2019, pp. 2928--2935.

\bibitem{ames}
A.~D. Ames, X.~Xu, J.~W. Grizzle, and P.~Tabuada, ``Control barrier function
  based quadratic programs for safety critical systems,'' \emph{IEEE
  Transactions on Automatic Control}, vol.~62, no.~8, pp. 3861--3876, 2016.

\bibitem{osqp}
B.~Stellato, G.~Banjac, P.~Goulart, A.~Bemporad, and S.~Boyd, ``{OSQP}: An
  operator splitting solver for quadratic programs,'' \emph{ArXiv e-prints},
  Nov. 2017.

\bibitem{quigley2009ros}
M.~Quigley, K.~Conley, B.~Gerkey, J.~Faust, T.~Foote, J.~Leibs, R.~Wheeler, and
  A.~Y. Ng, ``Ros: an open-source robot operating system,'' in \emph{ICRA
  workshop on open source software}, vol.~3, no. 3.2.\hskip 1em plus 0.5em
  minus 0.4em\relax Kobe, Japan, 2009, p.~5.

\end{thebibliography}

\end{document}